\documentclass[letterpaper,conference]{IEEEtran}
\usepackage{amssymb}
\usepackage{verbatim}
\usepackage{graphicx,epsfig}
\usepackage{amsfonts}
\usepackage{subfigure}
\usepackage[normalem]{ulem}
\usepackage{amsmath}
\usepackage{amsthm}
\usepackage[table]{xcolor}

\begin{document}

\title{Throughput Bound of XOR Coded Wireless Multicasting to Three Clients}

\author{Jalaluddin Qureshi, and Adeel Malik\\
Department of Electrical Engineering, Namal College, Mianwali, Pakistan.}

\maketitle

\begin{abstract}
It is a well-known result that constructing codewords over $GF(2)$ to minimize the number of transmissions for a single-hop wireless multicasting is an NP-complete problem. Linearly independent codewords can be constructed in polynomial time for all the $n$ clients, known as maximum distance separable (MDS) code, when the finite field size $q$ is larger than or equal to the number of clients, $q\geq n$. In this paper we quantify the exact minimum number of transmissions for a multicast network using erasure code when $q=2$ and $n=3$, such that $q<n$. We first show that the use of Markov chain model to derive the minimum number of transmissions for such a network is limited for very small number of input packets. We then use combinatorial approach to derive an upper bound on the exact minimum number of transmissions. Our results show that the difference between the expected number of transmissions using XOR coding and MDS coding is negligible for $n=3$.\let\thefootnote\relax\footnote{Published in the proceedings of IEEE CAMAD, Sept. 2015, Surrey, UK.}
\end{abstract}

\newtheorem{axiom}{Axiom}
\newtheorem{lemma}{Lemma}
\newtheorem{theorem}{Theorem}
\newtheorem{proposition}{Proposition}
\newtheorem{corollary}{Corollary}

\section{Introduction} \label{sect:introduction}
Wireless multicasting is an efficient method of disseminating common data to multiple clients. However the shared wireless channel is susceptible to packet losses due to collision and burst errors due to signal fading and channel noise. The automatic repeat request (ARQ) protocol of retransmitting a lost or corrupted packet when the access point (AP) does not receive a positive acknowledgement (ACK) from a client is not scalable for large multicast networks.

Erasure coding is a class of forward error correction (FEC) which corrects packet losses and burst errors. Instead of retransmitting packets using ARQ, in erasure coding lost packets are coded before transmission. Consider for example an AP multicasting packets $p_1$ and $p_2$ to clients $c_1$ and $c_2$. Client $c_1$ did not receive $p_1$, and $c_2$ did not receive $p_2$, instead of retransmitting these two packets, $p_1$ and $p_2$ can be XOR coded, and the codeword $p_1\oplus p_2$ is transmitted which both $c_1$ and $c_2$ can decode to recover the lost packets. 

Erasure coding has been proposed as an efficient method to improve reliability for various single-hop wireless multicast applications due to the broadcast nature of wireless transmission, such as in wireless sensor networks (WSNs) to dissemate software updates for debugging and task modifications~\cite{Dong11}, and multicasting over Wi-Fi~\cite{Chandra09}. It has also been adopted in various transmission standards~\cite{ETSI11}. The Raptor R10 code for instance has been adopted in third generation partnership project (3GPP) multimedia broadcast multicast service (MBMS) for multicasting file delivery and streaming applications, and the Reed-Solomon (RS) code is used in digital video broadcasting (DVB) for multicasting live video.

There is a special interest in erasure codes over $GF(2)$ as it involves the relatively simple operation of XOR addition for encoding and decoding, which minimizes the computation cost during encoding and decoding. It is due to this reason that the widely adopted Raptor R10 code and erasure coding for energy constrained WSN is constructed over $GF(2)$~\cite{Dong11}. However the problem of constructing coding vectors over $GF(2)$ to minimize the number of multicast transmissions is an NP-complete problem in general~\cite{Rouayheb07}. Many heuristic coding algorithms have been proposed to construct efficient erasure codes over $GF(2)$, a survey and comparison of which can be found in~\cite{Qureshi14}. 

When the finite field size is given as $q\geq n$, where $n$ is the number of clients in the network, it has been shown that a linearly independent codeword can be found for all the clients of the network in polynomial time~\cite{KChi10}, such a code is also known as maximum distance separable (MDS) code. The MDS code therefore serves as the lower bound of the expected number of transmissions for any linear erasure code, as each client need to receive exactly $k$ codewords before it can decode the $k$ input packets.

In this work-in-progress paper we quantify the exact minimum number of transmissions for a restricted class of wireless multicast network where $q<n$, given by $q=2$ and $n=3$. To the best of our knowledge this is the first work of its kind to study the throughput bound of wireless multicasting when coding vector for erasure coding is constructed non-randomly and over field size given as $q<n$. 

We first use Markov chain model, validated with simulations, to derive the exact minimum number of transmissions for $k\leq 3$ in Section~\ref{sec:markov}, where $k$ is the number of input packets. Due to increasing space of the number of Markov chain states, we then use combinatorial approach to derive an upper bound on the exact minimum number of transmissions for arbitrary $k$ in Section~\ref{sec:combinatorial}. Part of our future work is to generalize the result for any arbitrary number of clients $n$, and derive an upper bound which is lower than the bound presented in this paper, which we discuss along with the conclusion of our work in the Section~\ref{sec:conclusion}.

\section{Preliminaries}\label{sec:preliminaries}
\subsection{System Model and Notations}
Consider an access point (AP) multicasting $k$ input packets to $n$ clients. Packet loss at each client is assumed to be independent and identically distributed (iid), following the Bernoulli model with packet loss probability of $p$, $0\leq p<1$, and successful packet reception probability of $s=1-p$. The AP has the knowledge of the packets and codewords which each of the clients has received. In this paper we assume that $n=3$ and the finite field size over which an AP generates a coding vector is given by $q=2$. The set of input packets is given by $P$ and the set of codewords which client $c_i$ has received is given by $Y_i$.

The expected number of transmissions before all $n$ clients have $k$ linearly independent codewords is given by $E[t_x]$. The retransmission ratio $R_t$ is given by $\frac{E[t_x]}{k}$. The problem we consider is, what will be the exact minimum number of transmissions before all the $n$ clients receive $k$ linearly independent packets.

The matrix $H_i$, $H_i\in GF(2)^{r_i\times k}$, represents the coding coefficient matrix of the linearly independent codewords client $c_i$ has received. An input packet can be treated as a codeword with a coding vector given by a standard unit vector. The rank of the matrix $H_i$ is denoted by $r_i$. Once client $c_i$ has received $k$ linearly independent codewords, it can decode the $k$ input packets by the operation $H_i^{-1}\cdot C_i^T$, where $C_i$ is the vector of $k$ codewords. We call client $c_i$ satisfied if $r_i=k$, and unsatisfied otherwise. The number of unique codewords which are linearly dependent for at least one unsatisfied client is denoted by $T_d$. 

\subsection{Related Work}
The throughput bound for wireless multicasting has been studied for random linear (RL) code and MDS code. In RL coding, the coding coefficient is randomly and uniformly selected from $GF(q)$. In MDS coding the AP generates a codeword which is linearly independent for all the clients, hence each of the $n$ clients need to receive exactly $k$ codewords before it decodes the $k$ input symbols.

The expected number of transmissions for wireless multicasting using RL code and MDS code is derived by Sagduyu and Ephremides~\cite{Sagduyu07} and Ghaderi \textit{et al.}~\cite{Ghaderi07} respectively. For RL coding Lucani \textit{et al.}~\cite{Lucani09} derived the upper bound on the mean number of codewords which each client need to receive before decoding the $k$ input symbols.

\section{Markov Chain Approach}\label{sec:markov}
In this section we derive the expected number of transmission $E[t_x]$ assuming that the AP always transmits a codeword which is linearly independent for the maximum number of clients for $k\leq 3$. The case of $k=1$ is trivial which we exclude from our result. 

We first define various states for a given network parameters, and construct transition matrix. The variable $E[t_x]$ is calculated by the unique solution of the following equations,  

\begin{equation}
\mu_i=1+\sum_{j=1}^m a_{ij} \mu_j, \hspace{0.45cm} \textrm{for all transient states},
\end{equation}

\noindent where $a_{ij}$ is the probability of transition from state $S_i$ to state $S_j$, $m$ is the total number of states in the Markov chain, and $\mu_i=0$ for all recurrent states~\cite[Section 7.4]{Bertsekas08}. The Markov chain reaches absorbing state iff $r_i=k$ for all clients. We refer interested reader to the referenced book~\cite{Bertsekas08} for description on calculating $E[t_x]$ through this method. 

\begin{figure*}
\centering
\subfigure[$k=2$]{%
\centering
\includegraphics[width=0.31\textwidth]{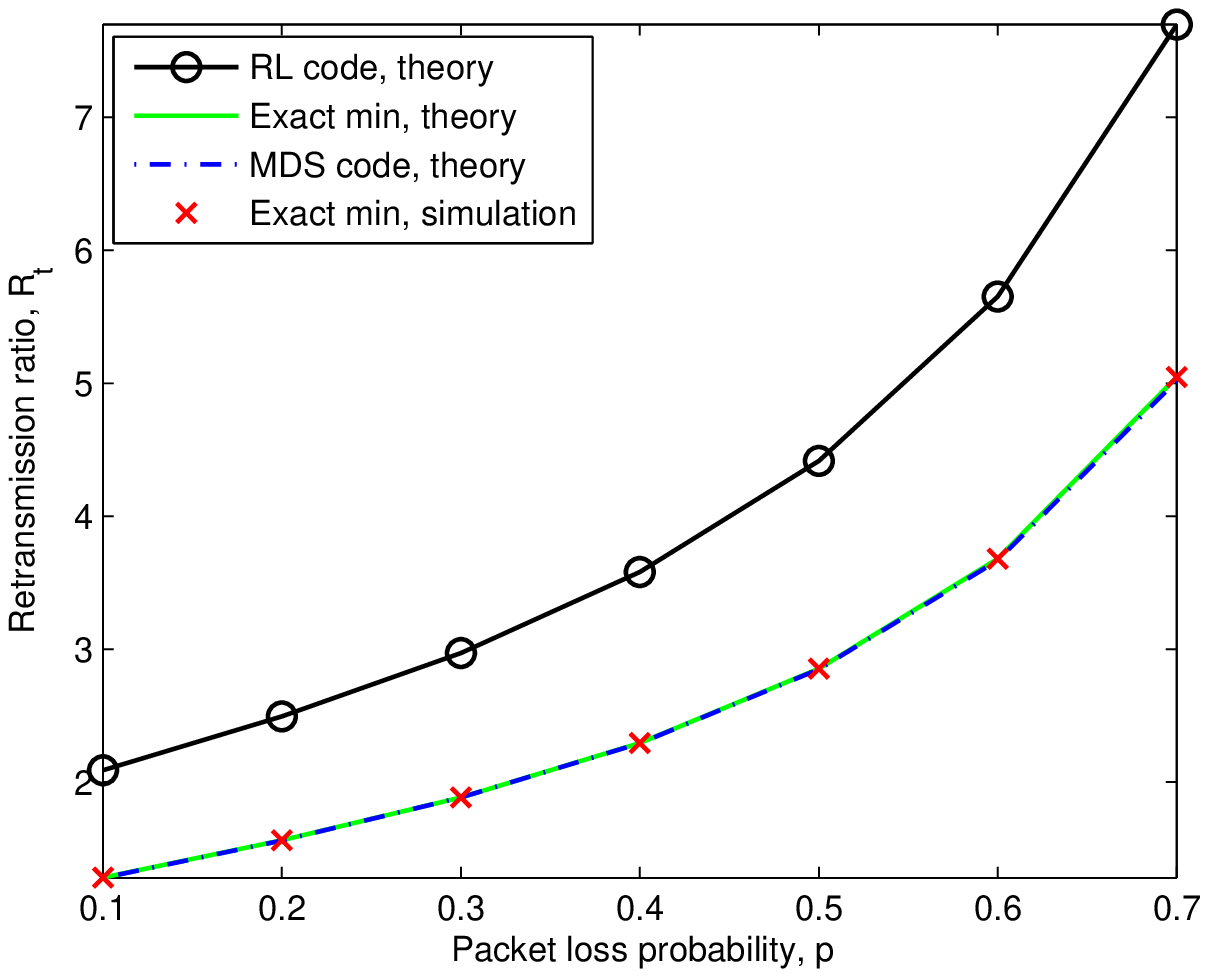}
\label{fig:subfigure1}}
\subfigure[$k=3$]{%
\centering
\includegraphics[width=0.31\textwidth]{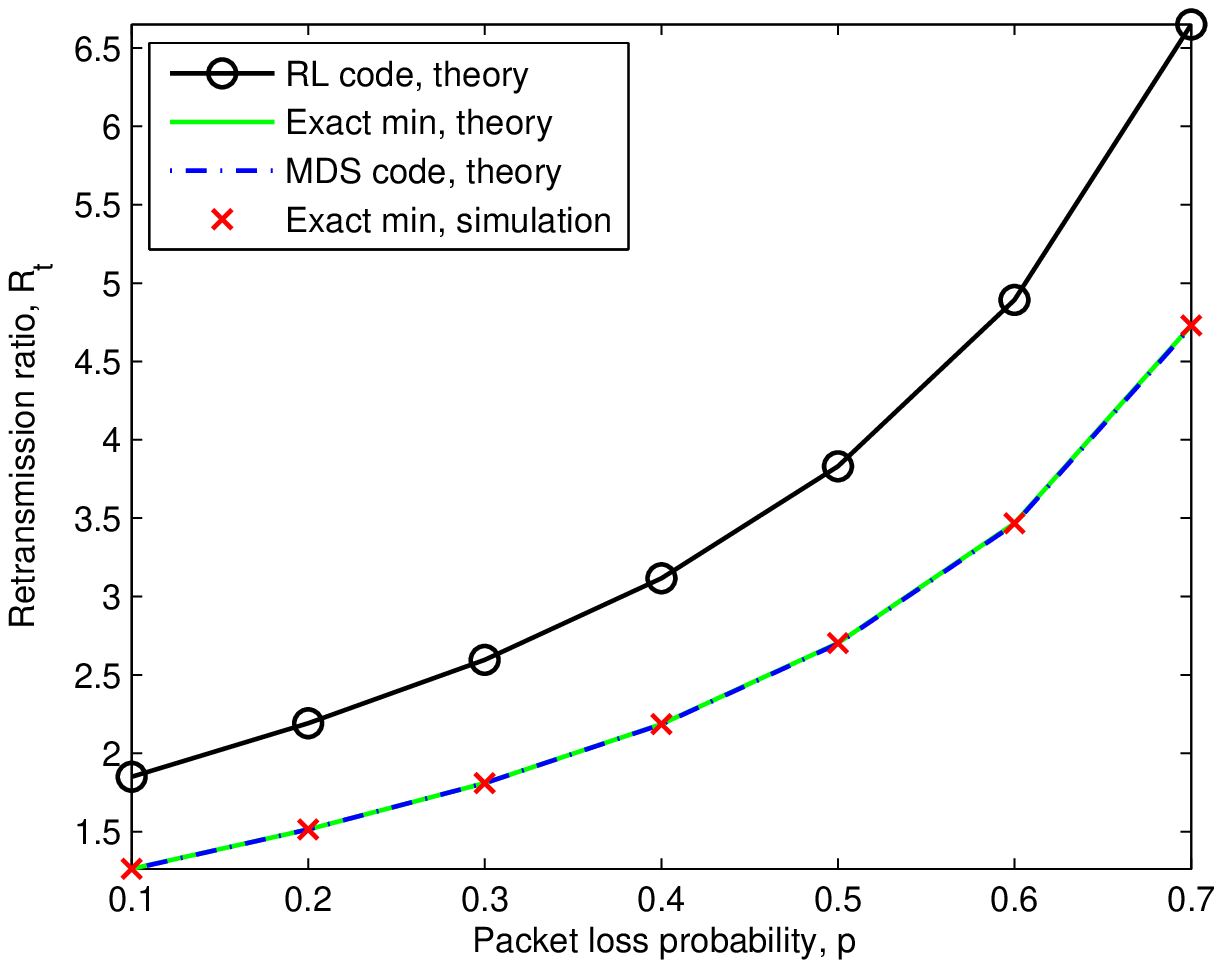}
\label{fig:subfigure2}}
\subfigure[Difference in $R_t$]{%
\centering
\includegraphics[width=0.32\textwidth]{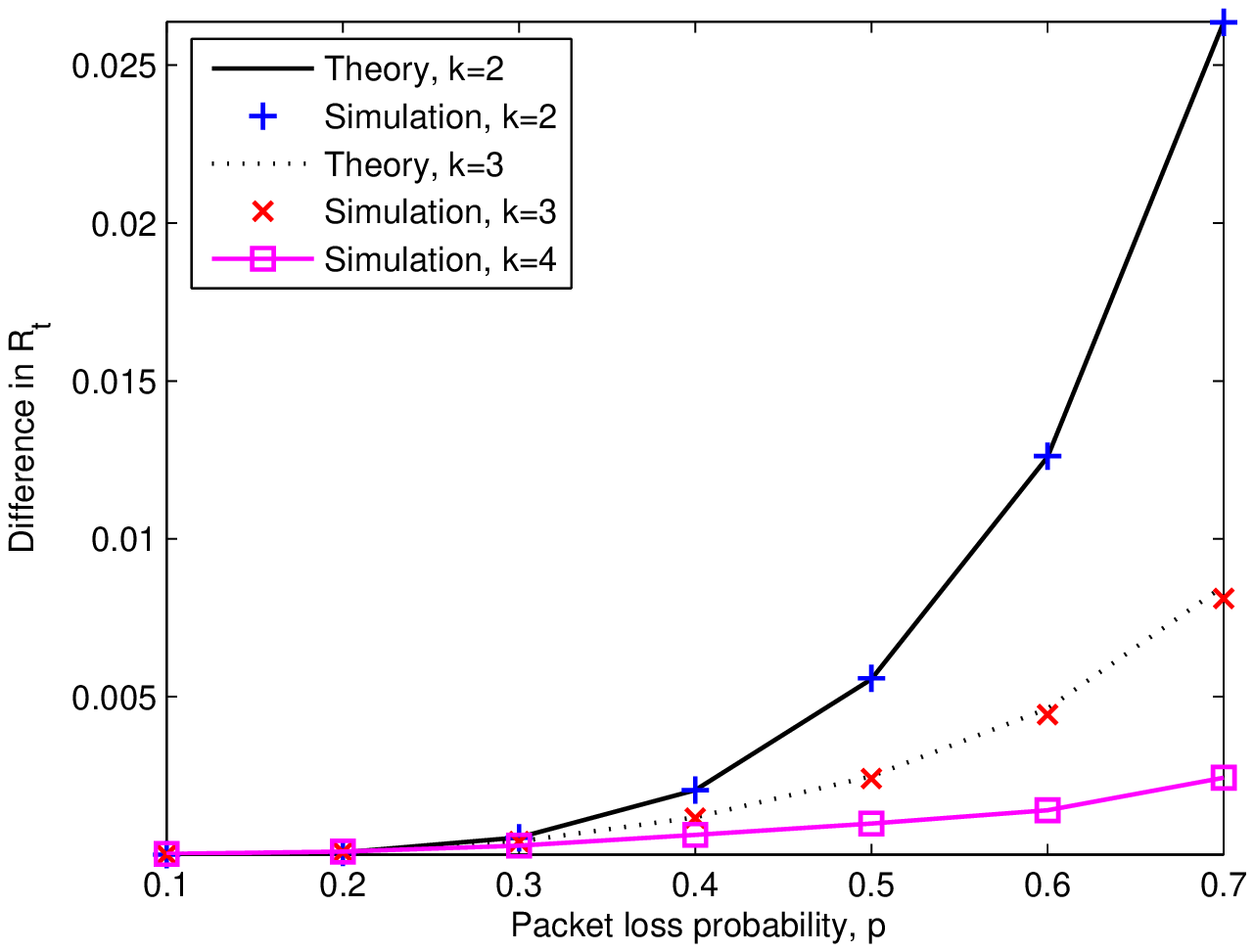}
\label{fig:subfigure3}}
\caption{Retransmission ratio for exact minimum number of transmission using $GF(2)$ erasure code derived using Markov chain model, compared with RL code over $GF(2)$ and MDS code for (a) $k=2$, (b) $k=3$, and (c) difference in $R_t$ of exact minimum transmission of $GF(2)$ erasure code and MDS code.}
\label{fig:Rt}
\end{figure*}

\subsection{Markov chain for $k=2$}
When $k=2$, the set of $\mathsf{span}(P)$ is given as $\{p_1, p_2, p_1\oplus p_2\}$. An AP cannot transmit a linearly independent codeword for all the three unsatisfied clients iff one of the client has $p_1$, another has $p_2$, and the third client has $p_1\oplus p_2$, which is represented by state $S_7$ in our Markov chain model. We distinguish $S_7$ from $S_5$ where an AP may be able to generate a linearly independent codeword, e.g. when all three clients have $p_1$. A self-explanatory description of the other states is given in Table~\ref{table:table_k2}, and a transition matrix is given in Table~\ref{table:matrix_k2}.


\begin{table}[t]
\begin{tabular}{|c|l}
\hline
\textbf{State} & \multicolumn{1}{c|}{\textbf{Description}} \\ \hline
$S_0$ & \multicolumn{1}{l|}{$r_i=0$ for all clients.} \\ \hline
$S_1$ & \multicolumn{1}{l|}{$r_i=1$ for one client, and $r_i=0$ for two clients.} \\ \hline
$S_2$ & \multicolumn{1}{l|}{$r_i=1$ for two clients, and $r_i=0$ for one client, and $T_d=1$.} \\ \hline
$S_3$ & \multicolumn{1}{l|}{$r_i=0$ for two clients, and $r_i=2$ for one client.} \\ \hline
$S_4$ & \multicolumn{1}{l|}{$r_i=1$ for two clients, $r_i=0$ for one client, and $T_d=2$.} \\ \hline
$S_5$ & \multicolumn{1}{l|}{$r_i=1$ for all clients, and $T_d\leq 2$.} \\ \hline 
$S_6$ & \multicolumn{1}{l|}{$r_i=2$, $r_i=1$, and $r_i=0$ for each of the clients.} \\ \hline
$S_7$ & \multicolumn{1}{l|}{$r_i=1$ for all clients, and $T_d=3$.} \\ \hline
$S_8$ & \multicolumn{1}{l|}{$r_i=2$ for one client, and $r_i=1$ for two clients.} \\ \hline
$S_9$ & \multicolumn{1}{l|}{$r_i=2$ for two clients, and $r_i=0$ for one client.} \\ \hline
$S_{10}$ & \multicolumn{1}{l|}{$r_i=2$ for two clients, and $r_i=1$ for one client.} \\ \hline
$S_{11}$ & \multicolumn{1}{l|}{$r_i=2$ for all clients. Absorbing state.} \\ \hline
\end{tabular}
\\
\caption{Markov states description for $k=2$.}
\label{table:table_k2}
\end{table}

\begin{table}[t]
\begin{frame}

\resizebox{\linewidth}{!}{%
$\displaystyle
\left( \begin{array}{cccccccccccc}
p^{3} &  3sp^{2} &  3s^{2}p &  0&  0&  s^{3}&  0&  0&  0&  0&  0 & 0  \\
0 &  p^{3}&  0&  sp^{2}&  2sp^{2}&  s^{2}p&  2s^{2}p&  0&  s^{3}&  0&  0  & 0  \\
0 &  0&  p^{3}&  0&  0&  sp^{2}&  2sp^{2}&  0&  2s^{2}p&  s^{2}p&  s^{3}  & 0 \\
0&  0&  0&  p^{2}&  0&  0&  2sp&  0&  s^{2}&  0&  0 & 0 \\
0&  0&  0&  0&  p^{3}&  0&  2sp^{2}&  sp^{2} &  2s^{2}p&  s^{2}p&  s^{3} & 0 \\
0&  0&  0&  0&  0&  p^{3}&  0&  0&  3sp^{2}&  0&  3s^{2}p & s^{3} \\
0&  0&  0&  0&  0&  0&  p^{2}&  0&  sp &  sp &  s^{2} & 0 \\
0&  0&  0&  0&  0&  0&  0&  p^{3}+sp^{2}&  2sp^{2}+2s^{2}p&  0&  s^{3}+s^{2}p & 0  \\
0&  0&  0&  0&  0&  0&  0&  0&  p^{2}&  0&  2sp  & s^{2}  \\
0&  0&  0& 0&  0&  0&  0&  0&  0&  p&  s& 0   \\
0 & 0 & 0 & 0 & 0 & 0 & 0 & 0 & 0 & 0 & p & s \\
0 & 0 & 0 & 0 & 0 & 0 & 0 & 0 & 0 & 0 & 0 & 1\\

\end{array} \right)$
}
\end{frame}
\caption{Transition matrix for $k=2$, element $a_{ij}$ represents transition probability from state $S_i$ to $S_j$.}
\label{table:matrix_k2}
\end{table}
\subsection{Markov chain for $k=3$}
When $k=3$, the set of $\mathsf{span}(P)$ is given as $\{p_1, p_2, p_3, p_1\oplus p_2, p_1\oplus p_3, p_2\oplus p_3, p_1\oplus p_2\oplus p_3\}$. An AP can not transmit a linearly independent codeword for three unsatisfied clients iff $r_i=2$ for all clients and $T_d=7$, representated by state $S_{27}$. A description of the other states is given in Table~\ref{table:table_k3}, and a transition matrix is given in Table~\ref{table:matrix_k3}.

We use $T_d$ to distinguish between different states with same rank $r_i$ distribution. State $S_{15}$ for example could represents a scenario when $c_1$ and $c_2$ has $p_1$, and $c_3$ has $p_2$ and $p_3$ resulting in $T_d=4$. While $S_{16}$ could represent a scenario when $c_1$ and $c_2$ has $p_1$, and $c_3$ has $p_1$ and $p_2$ resulting in $T_d=3$.

\begin{table}[h]
\begin{tabular}{|c|l}
\hline
\textbf{State} & \multicolumn{1}{c|}{\textbf{Description}} \\ \hline
$S_0$ & \multicolumn{1}{l|}{$r_i=0$ for all clients.} \\ \hline
$S_1$ & \multicolumn{1}{l|}{$r_i=0$ for two clients, and $r_i=1$ for one client.} \\ \hline
$S_2$ & \multicolumn{1}{l|}{$r_i=1$ for two clients, $r_i=0$ for one client, and $T_d=1$.} \\ \hline
$S_3$ & \multicolumn{1}{l|}{$r_i=2$ for one clients, and $r_i=0$ for two client.} \\ \hline
$S_4$ & \multicolumn{1}{l|}{$r_i=1$ for two clients, $r_i=0$ for one client, and $T_d=2$.} \\ \hline
$S_5$ & \multicolumn{1}{l|}{\begin{tabular}[c]{@{}l@{}}$r_i=2$, $r_i=1$, and $r_i=0$ for each of the clients\\ and $T_d=3$.  \end{tabular}} \\ \hline
$S_6$ & \multicolumn{1}{l|}{\begin{tabular}[c]{@{}l@{}}$r_i=2$, $r_i=1$, and $r_i=0$ for each of the clients\\ and $T_d=4$.  \end{tabular}} \\ \hline
$S_7$ & \multicolumn{1}{l|}{$r_i=0$ for two clients, and $r_i=3$ for one client.} \\ \hline
$S_8$ & \multicolumn{1}{l|}{$r_i=1$ for all clients, and $T_d=1$.} \\ \hline
$S_9$ & \multicolumn{1}{l|}{$r_i=1$ for all clients, and $T_d=2$.} \\ \hline
$S_{10}$ & \multicolumn{1}{l|}{$r_i=1$ for all clients and $T_d=3$.} \\ \hline
$S_{11}$ & \multicolumn{1}{l|}{$r_i=1$ for two clients, and $r_i=2$ for one client.} \\ \hline
$S_{12}$ & \multicolumn{1}{l|}{$r_i=2$ for two clients, $r_i=0$ for one client, and $T_d\neq 5$.} \\ \hline
$S_{13}$ & \multicolumn{1}{l|}{$r_i=3$, $r_i=1$, and $r_i=0$ for each of the clients.} \\ \hline
$S_{14}$ & \multicolumn{1}{l|}{$r_i=2$ for two clients, $r_i=0$ for one client and $T_d=5$.} \\ \hline
$S_{15}$ & \multicolumn{1}{l|}{$r_i=1$ for two clients, $r_i=2$ for one client and $T_d=4$.} \\ \hline
$S_{16}$ & \multicolumn{1}{l|}{\begin{tabular}[c]{@{}l@{}}$r_i=1$ for two clients, $r_i=2$ for one client\\ and $T_d=3$ or $T_d=5$.\end{tabular}} \\ \hline
$S_{17}$ & \multicolumn{1}{l|}{$r_i=2$ for two clients, and $r_i=1$ for one client.} \\ \hline
$S_{18}$ & \multicolumn{1}{l|}{$r_i=1$ for two clients, and $r_i=3$ for one client.} \\ \hline
$S_{19}$ & \multicolumn{1}{l|}{$r_i=3$, $r_i=2$, and $r_i=0$ for each of the clients.} \\ \hline
$S_{20}$ & \multicolumn{1}{l|}{\begin{tabular}[c]{@{}l@{}}$r_i=2$ for two clients, $r_i=1$ for one client\\ and $T_d=5$ or $T_d=6$.\end{tabular} } \\ \hline
$S_{21}$ & \multicolumn{1}{l|}{$r_i=2$ for all clients, and $2\leq T_d\leq 6$.} \\ \hline
$S_{22}$ & \multicolumn{1}{l|}{$r_i=3$, $r_i=2$, and $r_i=1$ for each of the clients.} \\ \hline
$S_{23}$ & \multicolumn{1}{l|}{$r_i=3$ for two clients, and $r_i=0$ for one client.} \\ \hline
$S_{24}$ & \multicolumn{1}{l|}{$r_i=2$ for all clients and $T_d=7$.} \\ \hline
$S_{25}$ & \multicolumn{1}{l|}{$r_i=2$ for two clients, and $r_i=3$ for one client.} \\ \hline
$S_{26}$ & \multicolumn{1}{l|}{$r_i=3$ for two clients, and $r_i=1$ for one client.} \\ \hline
$S_{27}$ & \multicolumn{1}{l|}{$r_i=3$ for two clients, and $r_i=2$ for one client.} \\ \hline
$S_{28}$ & \multicolumn{1}{l|}{$r_i=3$ for all clients. Absorbing state.} \\ \hline
\end{tabular}
\\
\caption{Markov states description for $k=3$.}
\label{table:table_k3}
\end{table}

\begin{table*}[t]
\begin{frame}

\resizebox{\linewidth}{!}{%
$\displaystyle
\left( \begin{array}{ccccccccccccccccccccccccccccc}
p^3 & 3sp^2 & 3s^2p & 0 & 0 & 0 & 0 & 0 & s^3 & 0 & 0 & 0 & 0 & 0 & 0 & 0 & 0 & 0 & 0 & 0 & 0 & 0 & 0 & 0 & 0 & 0 & 0 & 0 & 0 \\
0 & p^3 & 0 & sp^2 & 2sp^2 & 2s^2p & 0 & 0 & 0 & s^2p & 0 & 0 & 0 & 0 & 0 & 0 & s^3 & 0 & 0 & 0 & 0 & 0 & 0 & 0 & 0 & 0 & 0 & 0 & 0  \\
0 & 0 & p^3 & 0 & 0 & 2sp^2 & 0 & 0 & 0 & sp^2 & 0 & 2s^2p & s^2p & 0 & 0 & 0 & 0 & s^3 & 0 & 0 & 0 & 0 & 0 & 0 & 0 & 0 & 0 & 0 & 0  \\
0 & 0 & 0 & p^3 & 0 & 0 & 2sp^2 & sp^2 & 0 & 0 & 0 & s^2p & 0 & 2s^2p & 0 & 0 & 0 & 0 & s^{3} & 0 & 0 & 0 & 0 & 0 & 0 & 0 & 0 & 0 & 0  \\
0 & 0 & 0 & 0 & p^3 & 0 & 2sp^2 & 0 & 0 & 0 & sp^2 & 0 & 0 & 0 & s^2p & 2s^2p & 0 & 0 & 0 & 0 & s^3 & 0 & 0 & 0 & 0 & 0 & 0 & 0 & 0  \\
0 & 0 & 0 & 0 & 0 & p^3 & 0 & 0 & 0 & 0 & 0 & 0 & 0 & sp^2 & sp^2 & sp^2 & 0 & s^2p & s^2p & s^2p & 0 & 0 & s^3 & 0 & 0 & 0 & 0 & 0 & 0  \\
0 & 0 & 0 & 0 & 0 & 0 & p^3 & 0 & 0 & 0 & 0 & 0 & 0 & sp^2 & sp^2 & 0 & sp^2 & s^2p & s^2p & s^2p & 0 & 0 & s^3 & 0 & 0 & 0 & 0 & 0 & 0 \\
0 & 0 & 0 & 0 & 0 & 0 & 0 & p^2 & 0 & 0 & 0 & 0 & 0 & 2sp & 0 & 0 & 0 & 0 & s^2 & 0 & 0 & 0 & 0 & 0 & 0 & 0 & 0 & 0 & 0 \\
0 & 0 & 0 & 0 & 0 & 0 & 0 & 0 & p^3 & 0 & 0 & 0 & 0 & 0 & 0 & 0 & 3sp^2 & 3s^2p & 0 & 0 & 0 & s^3 & 0 & 0 & 0 & 0 & 0 & 0 & 0 \\
0 & 0 & 0 & 0 & 0 & 0 & 0 & 0 & 0 & p^3 & 0 & 0 & 0 & 0 & 0 & 3sp^2 & 0 & 3s^2p & 0 & 0 & 0 & s^3 & 0 & 0 & 0 & 0 & 0 & 0 & 0  \\
0 & 0 & 0 & 0 & 0 & 0 & 0 & 0 & 0 & 0 & p^3 & 0 & 0 & 0 & 0 & 2sp^2 & sp^2 & 3s^2p & 0 & 0 & 0 & s^3 & 0 & 0 & 0 & 0 & 0 & 0 & 0 \\
0 & 0 & 0 & 0 & 0 & 0 & 0 & 0 & 0 & 0 & 0 & p^3 & 0 & 0 & 0 & 0 & 0 & 2sp^2 & sp^2 & 0 & 0 & s^2p & 2s^2p & 0 & 0 & s^3 & 0 & 0 & 0 \\
0 & 0 & 0 & 0 & 0 & 0 & 0 & 0 & 0 & 0 & 0 & 0 & p^3 & 0 & 0 & 0 & 0 & sp^2 & 0 & 2sp^2 & 0 & 0 & 2s^2p & s^2p & 0 & 0 & s^3 & 0 & 0 \\
0 & 0 & 0 & 0 & 0 & 0 & 0 & 0 & 0 & 0 & 0 & 0 & 0 & p^2 & 0 & 0 & 0 & 0 & sp & sp & 0 & 0 & s^2 & 0 & 0 & 0 & 0 & 0 & 0 \\
0 & 0 & 0 & 0 & 0 & 0 & 0 & 0 & 0 & 0 & 0 & 0 & 0 & 0 & p^3 & 0 & 0 & 0 & 0 & 2sp^2 & sp^2 & 0 & 2s^2p & s^2p & 0 & 0 & s^3 & 0 & 0 \\
0 & 0 & 0 & 0 & 0 & 0 & 0 & 0 & 0 & 0 & 0 & 0 & 0 & 0 & 0 & p^3 & 0 & sp^2 & sp^2 & 0 & sp^2 & s^2p & 2s^2p & 0 & 0 & s^3 & 0 & 0 & 0 \\
0 & 0 & 0 & 0 & 0 & 0 & 0 & 0 & 0 & 0 & 0 & 0 & 0 & 0 & 0 & 0 & p^3 & 0 & sp^2 & 0 & 2sp^2 & s^2p & 2s^2p & 0 & 0 & s^3 & 0 & 0 & 0  \\
0 & 0 & 0 & 0 & 0 & 0 & 0 & 0 & 0 & 0 & 0 & 0 & 0 & 0 & 0 & 0 & 0 & p^3 & 0 & 0 & 0 & sp^2 & 2sp^2 & 0 & 0 & 2s^2p & s^2p & s^3 & 0  \\
0 & 0 & 0 & 0 & 0 & 0 & 0 & 0 & 0 & 0 & 0 & 0 & 0 & 0 & 0 & 0 & 0 & 0 & p^2 & 0 & 0 & 0 & 2sp & 0 & 0 & s^2 & 0 & 0 & 0\\
0 & 0 & 0 & 0 & 0 & 0 & 0 & 0 & 0 & 0 & 0 & 0 & 0 & 0 & 0 & 0 & 0 & 0 & 0 & p^2 & 0 & 0 & sp & sp & 0 & 0 & s^2 & 0 &  0\\
0 & 0 & 0 & 0 & 0 & 0 & 0 & 0 & 0 & 0 & 0 & 0 & 0 & 0 & 0 & 0 & 0 & 0 & 0 & 0 & p^3 & 0 & 2sp^2 & 0 & sp^2 & 2s^2p & s^2p & s^3 & 0 \\
0 & 0 & 0 & 0 & 0 & 0 & 0 & 0 & 0 & 0 & 0 & 0 & 0 & 0 & 0 & 0 & 0 & 0 & 0 & 0 & 0 & p^3 & 0 & 0 & 0 & 3sp^2 & 0 & 3s^2p & s^3 \\
0 & 0 & 0 & 0 & 0 & 0 & 0 & 0 & 0 & 0 & 0 & 0 & 0 & 0 & 0 & 0 & 0 & 0 & 0 & 0 & 0 & 0 & p^2 & 0 & 0 & sp & sp & s^2 & 0 \\
0 & 0 & 0 & 0 & 0 & 0 & 0 & 0 & 0 & 0 & 0 & 0 & 0 & 0 & 0 & 0 & 0 & 0 & 0 & 0 & 0 & 0 & 0 & p & 0 & 0 & s & 0 & 0 \\
0 & 0 & 0 & 0 & 0 & 0 & 0 & 0 & 0 & 0 & 0 & 0 & 0 & 0 & 0 & 0 & 0 & 0 & 0 & 0 & 0 & 0 & 0 & 0 & p^2(s+p) & 2sp(s+p) & 0 & s^2(s+p) & 0  \\
0 & 0 & 0 & 0 & 0 & 0 & 0 & 0 & 0 & 0 & 0 & 0 & 0 & 0 & 0 & 0 & 0 & 0 & 0 & 0 & 0 & 0 & 0 & 0 & 0 & p^2 & 0 & 2sp & s^2 \\ 
0 & 0 & 0 & 0 & 0 & 0 & 0 & 0 & 0 & 0 & 0 & 0 & 0 & 0 & 0 & 0 & 0 & 0 & 0 & 0 & 0 & 0 & 0 & 0 & 0 & 0 & p & s & 0 \\
0 & 0 & 0 & 0 & 0 & 0 & 0 & 0 & 0 & 0 & 0 & 0 & 0 & 0 & 0 & 0 & 0 & 0 & 0 & 0 & 0 & 0 & 0 & 0 & 0 & 0 & 0 & p & s \\
0 & 0 & 0 & 0 & 0 & 0 & 0 & 0 & 0 & 0 & 0 & 0 & 0 & 0 & 0 & 0 & 0 & 0 & 0 & 0 & 0 & 0 & 0 & 0 & 0 & 0 & 0 & 0 & 1\\

\end{array} \right)$
}
\end{frame}
\caption{Transition matrix for $k=3$, element $a_{ij}$ represents transition probability from state $S_i$ to $S_j$.}
\label{table:matrix_k3}
\end{table*}
\subsection{Analytical results}
Based on the solution of the set of equations $\mu_i$, a plot of $E[t_x]$ for exact minimum number of transmissions using XOR coded erasure code, compared with the analytical result of RL code over $GF(2)$ and MDS code adopted from~\cite{Ghaderi07, Sagduyu07}, is shown in Figure~\ref{fig:Rt}. For simulation carried to verify the correctness of our analytical result, the AP chooses a codeword which is linearly independent for maximum number of unsatisfied clients. When more than one such codewords exist which is linearly independent for maximum number of unsatisfied clients, we arbitrarily choose one of them.

In Figure~\ref{fig:subfigure1} and~\ref{fig:subfigure2} the $R_t$ of minimum number of transmissions almost overlaps the $R_t$ of MDS code. To highlight the difference between their performances, in Figure~\ref{fig:subfigure3} we plot difference between $R_t$ of exact minimum number of transmissions for $GF(2)$ erasure code and MDS code. The graphs shows that as $k$ increases the difference in $R_t$ decreases.

\section{Combinatorial Approach}\label{sec:combinatorial}
In this section we derive an upper bound on the minimum number of transmissions for $GF(2)$ erasure code. Our analysis makes use of the fact that the AP can generate a linearly independent codeword if there exist a codeword which is not given by the span of codewords of all unsatisfied clients. The cardinality of $\mathsf{span}(Y_i)$ is given as,
\begin{equation}
|\mathsf{span}(Y_i)|=\sum_{j=1}^{r_i} {r_i \choose j}=2^{r_i}-1.
\end{equation}

\noindent The AP can generate a linearly independent codeword if the following inequality is satisfied,
\begin{equation}\label{eq:inequality}
\sum_{i=1}^{n=3} (2^{r_i}-1) < |\mathsf{span}(P)|=2^k -1.
\end{equation}

\noindent We are interested to find the rank distribution $r_i$ under the constraint that $r_i<k, \forall r_i$. When $r_i=k$ for at least one of the client, then a linearly independent codeword can be found as the inequality $q\geq n$ will then be satisfied, where with slight abuse of notation $n$ is the number of unsatisfied clients. Based on inequality~\eqref{eq:inequality}, there exist a linearly independent codeword for all unsatisfied clients, if $r_i\leq k-2$ for two clients, and $r_i\leq k-1$ for another client.

\begin{lemma}\label{lemma}
For three unsatisfied clients, the only rank distribution where a linearly independent codeword over $GF(2)$ for all clients can not necessarily be generated is given by $r_i=k-1, \forall r_i$. 
\end{lemma}
\begin{proof}
From the result of inequality~\eqref{eq:inequality} we know that a linearly independent codeword for all clients exist if $r_i\leq k-2$ for two clients, and $r_i\leq k-1$ for another client. We now propose an algorithm which can find a linearly independent codeword for all clients if $r_i=k-1$ for two clients, and $r_i=k-2$ for another client. Our algorithm makes use of the fact that for a matrix, the rank of its rows is equal to the rank of its columns.  Without loss of generality consider that $r_1=k-1$, $r_2=k-1$ and $r_3=k-2$. Denote by $\mathbf{w}=[x_1, x_2, \ldots, x_k]$, $x_i\in GF(2)$, the coding vector we wish to construct. Append $\mathbf{w}$ as the last row of $H_i$.

For $H_1$ and $H_2$ perform column additions such that there exist two columns in each matrix which are linearly dependent while ignoring the last row, we label these two columns in $H_1$ as $a$ and $b$, and in $H_2$ as $c$ and $d$. Similarly, column additions are performed for $H_3$ such that there exist three linearly dependent columns while ignoring the last row, which we label as $e, f$ and $g$. 

Then vector $\mathbf{w}$ is linearly independent for all $H_i$ if columns $a$ and $b$ are not linearly dependent in $H_1$, $c$ and $d$ are not linearly dependent in $H_2$, and if any two of the following columns $e, f$ and $g$ are not linearly dependent in $H_3$. Denote by $S$ the set of $k$ columns in $H_i$, then this results in the set of Equations~\eqref{eq:H_1}-\eqref{eq:H_3iii}, where $\alpha_i$ is equal to 1 if the $i^{th}$ column was added to the $u^{th}$ column such that columns $u$ and $v$ while ignoring the last row are dependent in $H_i$, and 0 otherwise. Vector $\mathbf{w}$ is linearly independent for all clients if Equations~\eqref{eq:H_1},~\eqref{eq:H_2}, and either~\eqref{eq:H_3i} or~\eqref{eq:H_3ii} or~\eqref{eq:H_3iii} are satisfied. It is intuitive to see that due to sufficient degree of freedom an assignment of $x_i$ such that $\mathbf{w}$ is linearly independent for all clients exist. 

We now show that when $r_i=k-1, \forall r_i$, a linearly independent codeword for all clients cannot necessarily be generated. Without performing any column additions, and while ignoring the last row, columns $a$ and $b$ are linearly dependent in $H_1$, $c$ and $d$ are linearly dependent in $H_2$, and $e$ and $f$ are linearly dependent in $H_3$. Further if $a=f, b=d,$ and $c=e$, then there does not exist an assignment of $x_i$ such that $\mathbf{w}$ is linearly independent for all clients. This completes the proof.

\begin{subequations}
\begin{align}
x_a\oplus \bigoplus_{\forall i : i\in S-\{a,b\}} \alpha_i x_i \neq x_b, \label{eq:H_1}\\
x_c\oplus \bigoplus_{\forall i : i\in S-\{c,d\}} \alpha_i x_i \neq x_d, \label{eq:H_2}\\
x_e\oplus \bigoplus_{\forall i : i\in S-\{e,f\}} \alpha_i x_i \neq x_f, \label{eq:H_3i}\\
x_e\oplus \bigoplus_{\forall i : i\in S-\{e,g\}} \alpha_i x_i \neq x_g, \label{eq:H_3ii}\\
x_f\oplus \bigoplus_{\forall i : i\in S-\{f,g\}} \alpha_i x_i \neq x_g, \label{eq:H_3iii}
\end{align}
\end{subequations}
\end{proof}

Therefore based on Lemma~\ref{lemma}, considering that a maximum of one client receives linearly dependent codewords when $r_i=k-1, \forall r_i$, then that client need to receive $k+\delta$ codewords, while the other two clients need to receive $k$ codewords before all clients are satisfied.

We derive the expected value of $\delta$ assuming that the network enters a state where $r_i=k-1, \forall r_i$ such that the AP can not transmit a linearly independent codeword for all clients. Then the probability $P[\delta=\beta]$ that one of the client receives $\beta$ redundant codewords is given by,
\begin{equation}
P[\delta=\beta]=(sp^2)^{\beta-1}(sp^2+2s^2p+s^3),
\end{equation}

\noindent and the expected value of $\delta$ is given as,
\begin{equation}
\begin{split}
E[\delta] & = \sum_{\beta=1}^\infty \beta P[\delta=\beta] \\
& = (sp^2+2s^2p+s^3) \sum_{\beta=1}^\infty \beta (sp^2)^{\beta-1} \\
					& = \frac{sp^2+2s^2p+s^3}{(1-sp^2)^2},
\end{split}
\end{equation}

\noindent where $0<E[\delta]\leq 1$ for $0\leq p<1$. As our objective is to derive an upper bound on the exact minimum number of transmission, we quantify the expected number of transmissions such that two clients receive $k$ codewords and another client receive $k+1$ codewords.

Let $\ell$ denote the number of transmissions before all client receives at least $k$ linearly independent codewords. Denote by $D(m)$ the probability that $\ell\leq m$, i.e. $P[\ell\leq m]$. The probability that a client receives $j$ codewords out of $m$ transmissions follows Bernoulli distribution and is given as, $P[X=j]$. The probability that the client receives at least $k$ linearly independent codewords after $m$ transmissions is given as,
\begin{equation}
D_1(m)=\sum_{j=k}^m {m \choose j} s^j p^{m-j},
\end{equation}

\noindent and the probability that the client receives at least $k$ linearly independent codewords and one dependent codeword (i.e. $k+1$ codewords) after $m$ transmissions is given as,
\begin{equation}
D_2(m)=\sum_{j=k+1}^m {m \choose j} s^j p^{m-j}.
\end{equation}

\noindent As the codeword reception at each client is independent, probability $D(m)$ is given as, $D_1(m)^2\cdot D_2(m)$. The expected number of transmissions to transmit $k$ codewords to two clients and $k+1$ codewords to another clients is given as,
\begin{equation}
\begin{split}
E[\ell] & = \sum_{m=0}^\infty (1-D(m)) \\
        & = k+1+\sum_{m=k+1}^\infty (1-D(m)).
\end{split}
\end{equation}

\begin{figure}[t]
	\centering
		\includegraphics[width=0.50\textwidth]{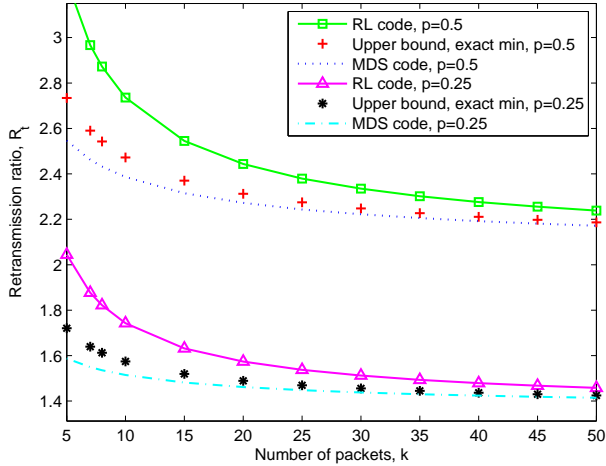}
	\caption{Retransmissions ratio $R_t$ of the upper bound of minimum number of $GF(2)$ coded wireless multicasting to three clients for $p=0.25$ and $p=0.5$ compared with RL coding over $GF(2)$ and MDS coding.}
	\label{fig:upper_bound}
\end{figure}

The quantity $E[\ell]$ is the upper bound of $E[t_x]$ for minimum number of transmissions using $GF(2)$ erasure coding. A plot of $E[\ell]$ compred with RL coding over $GF(2)$ and MDS code for $p=0.25$ and $p=0.5$ is given in Figure~\ref{fig:upper_bound}. Our result show that the value of $E[\ell]$ is very close to the expected number of transmissions using MDS code (which is also the lower bound of the expected number of transmission for any linear erasure code for wireless multicasting).

\section{Conclusion}\label{sec:conclusion}
In this paper we studied the problem of finding the minimum number of expected transmissions for wireless multicasting to three clients using XOR coded transmissions. For such a network a linearly independent codeword for all clients can not necessarily be generated by the transmitter. Our Markov chain analysis result showed that the difference between the minimum number of expected XOR coded transmission and expected transmission using MDS coding is negligible even for small packet batch size. Using combinatorial analysis we then showed that when using XOR coding, two of the clients need to receive $k$ codewords, while another client need to receive an expected $k+1$ codewords at most before all clients are satisfied.

We then derived the expected number of transmissions before two clients receive $k$ codewords and another client receive $k+1$ codewords. This expected number of transmission then served as the upper bound on the minimum number of transmissions using XOR coded multicasting to three clients. Our result on the upper bound of the minimum number of transmissions along with exact minimum number of transmission for $k\leq 3$ using Markov chain showed that the difference between minimum number of expected transmissions using XOR coding and expected number of transmissions using MDS code is negligible. Part of our future work is to generalize the results of our paper for an arbitrary number of clients.

\bibliographystyle{IEEEtranS}
\bibliography{IEEEabrv,mainJ}

\end{document}